\documentclass[conference, letterpaper]{IEEEtran}

\usepackage{url}
\usepackage{ifthen}
\usepackage{cite}
\usepackage[cmex10]{amsmath} 

\usepackage{balance}
\usepackage{amsmath}
\usepackage{amsthm}
\usepackage{amssymb}
\usepackage{amsfonts}
\usepackage{graphicx}
\usepackage{algorithm}
\usepackage{url}
\usepackage{lipsum}
\usepackage{booktabs,makecell,multirow}
\usepackage{hyperref}

\usepackage{latexsym}
\usepackage{color}
\usepackage{graphicx}
\usepackage{threeparttable}
\usepackage{float}

\usepackage{algpseudocode}
\usepackage{subcaption}

\errorcontextlines\maxdimen

\DeclareGraphicsExtensions{.tif,.pdf,.jpeg,.png,.jpg,.bmp,.svg,.eps,.EMF}

\newtheorem{theorem}{Theorem}
\newtheorem{definition}{Definition}

\newcommand{\cbc}{${\sf{CBC}}(n,m,k,t)$}
\newcommand{\ecbc}{${\sf{ECBC}}(n,m,k,t,r)$}

\IEEEoverridecommandlockouts

\begin{document}

\title{New Results on Erasure Combinatorial Batch Codes\\[.4ex]
  {\normalfont\large
	Phuc-Lu Le\IEEEauthorrefmark{1}, Son Hoang Dau\IEEEauthorrefmark{2},   Hy Dinh Ngo\IEEEauthorrefmark{1}, Thuc D. Nguyen\IEEEauthorrefmark{1}\IEEEauthorrefmark{4}}\\[-1.5ex]\IEEEcompsocitemizethanks{\IEEEcompsocthanksitem\IEEEauthorrefmark{4}Corresponding author} 
}

\author{\IEEEauthorblockA{\IEEEauthorrefmark{1}Faculty of Information Technology, Ho Chi Minh City University of Science, Vietnam \\
\IEEEauthorrefmark{2}The School of Computing
Technologies, RMIT University, Australia
\\ \{lplu, ndthuc, ndhy\}@fit.hcmus.edu.vn, \, sonhoang.dau@rmit.edu.au}}

\maketitle
\thispagestyle{plain}
\pagestyle{plain}

\pagenumbering{gobble}

\IEEEpeerreviewmaketitle

\begin{abstract} 
We investigate in this work the problem of Erasure Combinatorial Batch Codes, in which $n$ files are stored on $m$ servers so that every set of $n-r$ servers allows a client to retrieve at most $k$ distinct files by downloading at most $t$ files from each server. Previous studies have solved this problem for the special case of $t=1$ using Combinatorial Batch Codes.
We tackle the general case $t \geq 1$ using a generalization of Hall's theorem. Additionally, we address a realistic scenario in which the retrieved files are consecutive according to some order and provide a simple and optimal solution for this case.
\end{abstract}


\section{Introduction}


\label{sec:intro}



Combinatorial Batch Codes (CBC) were defined by Paterson {\it et al.} in \cite{paterson2009combinatorial} as a combinatorial version of batch codes (introduced by Ishai {\it et al.}~\cite{ishai2004batch}). More specifically, a CBC allows ones to store $n$ files on $m$ servers (possibly with repetitions among servers) such that any $k$ distinct items can be retrieved by downloading at most $t$ items from each server. The goal is to minimize the total number $N$ of items stored across all servers, given $n$, $m$, $k$ and $t$. To address the practical scenario in distributed storage systems where servers may fail or maintenance frequently, CBC was also generalized to Erasure Combinatorial Batch Codes (ECBC) by J. Jung {\it et al.} ~\cite{jung2018erasure}, which allows up to $r$ servers to fail. 
Several constructions of ECBC in which every item is stored in the same number of servers were developed in~\cite{jung2018erasure}. 

As far as we know, most of the previous works on CBC/ECBC assumed that $t = 1$, i.e., each server can be used at most $t=1$ time. The problem has been approached using various mathematical structures, including dual systems, extremal hypergraphs, and transversal matroids (see, e.g., \cite{bujtas2011optimal,balachandran2014extremal,brualdi2010combinatorial}). A number of works focused on solving CBC for specific cases, e.g. when $k$ is small ($k=3,4,5$) or when $n=m+2$~ \cite{jia2019some}. In almost all cases, the main used method is Hall's theorem from graph theory.
More specifically, by viewing a CBC/ECBC as a matching in a bipartite graph, one can apply Hall's condition to find the optimal total storage.

In this paper, using also Hall's theorem, we extend the previous studies to address the more general case of ECBC for any $r \ge 0$ and $t \ge 1$. 
Note that in practice, allowing the client to download more than one item from each server can improve the efficiency of data retrieval. 
Additionally, motivated by the application in video streaming where several consecutive video chunks are often downloaded at one time, we consider and provide the optimal solution for a special version of ECBC in which the retrieved files are consecutively indexed (instead of being random). Although this version is rather simple to address, it opens up an entirely new research direction on batch codes where batches to be downloaded are not randomly selected but follow some special patterns. Along with the approaching CBC/ECBC problems, we focus on the efficient algorithm to retrieve a list of files from the given ECBC system. Finally, we briefly note the relationship between ECBC and Non-Adaptive Group Testing (NAGT). 

\section{Preliminaries}

\subsection{Problem Definitions}
\label{sec:intro:prob}

We first recall the definition of an ECBC.

\begin{definition} 
\label{def:ECBC}
An Erasure Combinatorial Batch Code denoted by \ecbc{}  is a set system $(X, \mathcal{S})$ where $X = \{1,2,\ldots,n \}$ and $\mathcal{S}$ is a collection of $m$-subsets of $X$, denoted by $S_1,S_2,\ldots,S_m$, such that for every subset $X' \subset X$ of size up to $k$ and for every subset $J$ of $\{1,2,\ldots,m \}$ with $|J| \ge m-r$, there exists a subset $C_j \subset S_j$, $j \in J$, satisfying $|C_j| \le t$ and $X' =  \bigcup\limits_{j\in J}{{{C}_{j}}}.$
The goal is to construct an \ecbc{} that minimizes the total storage $N \triangleq \sum_{j=1}^m |B_j|$. 
\end{definition}

In Definition~\ref{def:ECBC}, the set $X$ can be considered as the list of indices of $n$ files while $S_1, S_2, \ldots, S_m$ are the indices of files that stored on $m$ server. Here, $r$ refers to the maximum number of inactive servers simultaneously.

Next, we recall the well-known Hall's Marriage theorem, which plays an essential role in the constructions of CBC/ECBC.

\subsection{Hall's Marriage Theorem}
Hall's theorem is stated as follows: for positive integers $m$ and $n$, assume that ${{A}_{1}},{{A}_{2}},\ldots ,{{A}_{n}}$ are $n$ subsets of $X=\{1,2,\ldots ,m\}$. If for every subset $I$ of $\{1,2,\ldots ,n\}$, the condition $\left| \bigcup\limits_{i\in I}{{{A}_{i}}} \right|\ge \left| I \right|$ holds, then there exist $n$ distinct elements from these $n$ subsets, one from each set. We refer to the aforementioned condition as the Hall's condition. 
The theorem states that in a simple undirected bipartite graph $G=(X,Y,E)$ with two disjoint sets of vertices $X,Y$ and $E$ as the edges connecting vertices between them, there exists a matching from $X$ to $Y$ if and only if, for any subset $X'$ of $X$ with size $k$, $\left| V(X')\right|\ge \left|X' \right|$, where $V(X')$ represents the set of vertices in $Y$ that are adjacent to some vertex in $X'$. 

In other words, Hall's condition can be applied to a subset $X'$ of $X$ with size $k$ to find a matching from $X'$ to $Y$. This is the other form of Hall’s theorem which is compatible with solving the CBC-ECBC problems and was mentioned by Bujtas and Tuza in ~\cite{bujtas2012relaxations}.

The way to apply Hall’s condition is straightforward: by defining the equivalence between the original problem and Hall's condition, and then establishing the dual condition, we can find the optimal solution for some cases of $m$ and $n$. The dual condition can be understood as the requirement that `the files must be stored on separate servers, so as to prevent servers from storing too many files'. These concepts form the basis for constructing the proof of the general case where $t \ge 1$ and $r \ge 0$ that will be discussed in more detail in the next section.

\subsection{Existing Results for CBC and ECBC}

The incidence matrix $A=(a_{ij})_{m\times n}$ of a \cbc{} or \ecbc{} is defined as: $a_{ij}=1$ if and only if server $i$ stores file $j$, and $a_{ij}=0$ otherwise.
Let $F_j\triangleq \{i : a_{ij} = 1\}$ be the set of 
rows that has a `1' in column $j$ and $S_i = \{j : a_{ij} = 1\}$ be the set of columns that has a `1' in row $i$. In other words, $F_j$ consists of the indices of servers that store file $j$, $j=1,\ldots,n$ and $S_i$ consists of the indices of files that are stored in server $i$, $i = 1,\ldots,m$. 


For some special parameter ranges, the minimum total store of a \cbc{} can be determined.

\begin{theorem}[Paterson, Stinson, and Wei~\cite{paterson2009combinatorial}]
Let $N(m,n,k)$ denote the minimum total storage of a \cbc{} when $t=1$. Then the following statements hold.
\begin{enumerate}
\item If $n=m$ then $N(n,m,k)=n$.
\item If $m=k$ then $N(n,m,k)=kn-k(k-1).$
\item If $n=m+1$ then $N(n,m,k)=m+k.$
\item If $n \ge (k-1)\binom{m}{k-1}$ then $$N(n,m,k)=kn-(k-1)\binom{m}{k-1}.$$
\item If $\binom{m}{k-2} \le n \le (k-1)\binom{m}{k-1}$ then $$N(n,m,k)=n(k-1)-\left\lfloor \frac{(k-1)\binom{m}{k-1}-n}{m-k+1} \right\rfloor.$$
\end{enumerate}    
\end{theorem}



Paterson, Stinson, and Wei~\cite{paterson2009combinatorial} were the first to study the property of the incidence matrix of a CBC ($t=1$) using the Hall's condition. 

\begin{theorem}[Paterson, Stinson, and Wei~\cite{paterson2009combinatorial}] The $m \times n$ binary matrix $A$ represents a CBC$(n,m,k)$ if and only if one of the following conditions holds.
\begin{enumerate} 
\item[(1)] For every $c \in \{1,2,\ldots ,k\}$, and for every subset $J\subset \{1,2,\ldots ,n\}$ of $c$ elements, it holds that $\left| \bigcup\limits_{j\in J}{{{F}_{j}}} \right|\ge c$. In other words, this condition requires that every set of $c$ files must be collectively stored on at least $c$ servers.
\item[(2)] For every $d\in \{0,1,\ldots ,k-1\}$, and every subset $I\subset \{1,2,\ldots ,m\}$ of size $d$, it holds that $\left| \bigcup\limits_{i\in I}{{{S}_{i}}} \right|\le d$. In other words, this condition requires that every set of $d$ servers store at most $d$ files.
\end{enumerate}
\end{theorem}


Bujtas and Tuza ~\cite{bujtas2012relaxations} extended 
the aforementioned approach to the CBC problem with an arbitrary $t$. 

\begin{theorem}[Bujtas and Tuza~\cite{bujtas2012relaxations}] The $m \times n$ binary matrix $A$ represents a \cbc{} if and only if one of the following conditions holds.
\begin{enumerate}
\item[(1)] For every $c\in \{1,2,\ldots ,k\},$ and every subset $J\subset \{1,2,\ldots ,n\}$ of size $c$, it holds that $\left| \bigcup\limits_{j\in J}{{{F}_{j}}} \right|\ge \left\lceil \frac{c}{t} \right\rceil$. In other words, the condition requires that any set of $c$ files must be collectively stored on at least $\left\lceil \frac{c}{t} \right\rceil $ servers.
\item[(2)] For every $d\in \{0,1,\ldots ,\left\lceil \frac{k}{t} \right\rceil -1\},$ and every subset $I\subset \{1,2,\ldots ,m\}$ of size $d$, it holds that $\left| \bigcup\limits_{i\in I}{{{S}_{i}}} \right|\le dt$. In other words, this condition requires that any set of $d$ servers together store at most $dt$ files.
\end{enumerate}
\end{theorem}


Jung, Mummert, Niese, and Schroeder~\cite{jung2018erasure} 
extended the previous results to the ECBC problem with $t=1$, via an extension of Hall's theorem. 

\begin{theorem}[Jung, Mummert, Niese \& Schroeder~\cite{jung2018erasure}]
The $m \times n$ binary matrix $A$ represents an \ecbc{} with $t = 1$ if and only if one of the following conditions holds.
\begin{enumerate}
\item[(1)] For every $c\in \{1,2,\ldots ,k\},$ and every subset $J\subset \{1,2,\ldots ,n\}$ of size $c$, it holds that $\left| \bigcup\limits_{j\in J}{{{F}_{j}}} \right|\ge r+c$. In other words, every set of $c \le k$ files collectively must be stored on at least $r+c$ servers.
\item[(2)]  For every $d\in \{r,r+1,\ldots ,r+k-1\},$ and for every subset $I\subset \{1,2,\ldots ,m\}$ of size $d$, it holds that $\left| \bigcup\limits_{i\in I}{{{S}_{i}}} \right|\le d-r$. In other words, every set of $d$ servers store at most $d-r$ files.
\end{enumerate}
\end{theorem}

With the background and definitions established in this section, we are now ready to present our main results.

\section{Main results}
\label{sec:results}

In this section, we establish the generalization of Hall’s condition for ECBC for any $r \ge 0$, $t \ge 1$, and demonstrate its application in determining the minimum total storage $N(n,m,k,r,t)$ of an ECBC problem in several cases. 

\subsection{The Main Theorem for ECBC with General $r$ and $t$}

\begin{theorem}
\label{thr:main}
A binary matrix $A$ of size $m \times n$ represents an \ecbc{} if and only if one of the following conditions holds.
\begin{enumerate}
\item[(1)] For every $c\in \{1,2,\ldots ,k\},$ and for every subset $J\subset \{1,2,\ldots ,n\}$ of size $c$, it holds that
$$\left| \bigcup\limits_{j\in J}{{{F}_{j}}} \right|\ge \left\lceil \frac{c}{t} \right\rceil +r.$$ In other words, the condition requires that any set of $c \le k$ files must be collectively stored on at least $\left\lceil \frac{c}{t} \right\rceil +r$ servers. 
\item[(2)] For every $d\in \{r,r+1,\ldots ,r+\Delta -1\},$ and for every subset $I\subset \{1,2,\ldots ,m\}$ of size $d$, it holds that 
$$\left| \bigcup\limits_{i\in I}{{{S}_{i}}} \right|\le t(d-r).$$ In other words, the condition requires that any set of $d$ servers store at most $t(d-r)$ files. Here, $\Delta \triangleq \left\lceil \frac{k}{t} \right\rceil$.

\end{enumerate}
\end{theorem}

\begin{proof} We divide the proof into two  steps as follows.

\textbf{First step.} We show that a binary matrix $A$ represents an \ecbc{} if and only if (1) holds. 

First, let us assume that the matrix $A$ represents an \ecbc{}. Regardless of $r$ unavailable servers, any $c\le k$ files can be retrieved from the remaining servers by downloading at most $t$ files from each. These files must be stored on at least $r+\left\lceil \frac{c}{t} \right\rceil $ servers originally, for otherwise, there would be at most $\left\lceil \frac{c}{t} \right\rceil -1$ available servers storing any of these files, and downloading $t$ from each would not be enough to recover $c$ files. 

Next, assume that the binary matrix $A$ meets condition (1) and moreover, there is a collection of unavailable servers $S$ with $\left| S \right|\le r.$ For every set $J\subset \{1,2,\ldots ,n\}$ of size $c$, i.e., $c$ files, the number of servers that store at least one of these files is 
\[\begin{aligned} \left| \bigcup\limits_{j\in J}{({{F}_{j}}\backslash S)} \right| & =\left| \left( \bigcup\limits_{j\in J}{{{F}_{j}}} \right)\backslash S \right|\ge \left| \bigcup\limits_{j\in J}{{{F}_{j}}} \right|-r \\
& \ge \left\lceil \frac{c}{t} \right\rceil +r-r=\left\lceil \frac{c}{t} \right\rceil. \end{aligned} \]


We now create an $(mt)\times n$ matrix $A'$ from $A$ by replicating each row of $A$ $t$ times. In other words, we replicate each server $t$ times, where each copy stores the same set of files as the original one. 
Then for the same set $J\subset \{1,2,\ldots ,n\}$ of size $c \le k$, the number of servers (including the duplicated ones) storing at least one file is 
$$t\left| \bigcup\limits_{j\in J}{({{F}_{j}}\backslash S)} \right|\ge \left\lceil t\cdot \frac{c}{t} \right\rceil =c,$$
which means that the Hall's condition is satisfied for the sets $F'_j$, $j=1,\ldots,n$, defined for the matrix $A'$.
By applying Hall's theorem, there is a matching between the $c$ files and a set of available servers (original and copies). Since each of original (available) servers appears at most $t$ times in this set, it appears in the matching at most $t$ times. This means that there are at most $t$ files downloaded from each available servers while retrieving $c$ files. Thus, the matrix $A$ represents an \ecbc{}.  

\textbf{Second step.} We aim to prove the equivalence between the two conditions (1) and (2). 

First, suppose that $A$ does not satisfy condition $(1)$. Then there exists $J\subset \{1,2,\ldots ,n\}$ of size $c \le k$ such that $\bigcup\limits_{j\in J}{{{F}_{j}}}$ contains $\left\lceil \frac{c}{t} \right\rceil +r-1$ elements at most. Take $I \triangleq \bigcup\limits_{j\in J}{{{F}_{j}}} \subset \{1,2,\ldots ,m\}$. 
Then we have
$$\left| I \right|\le \left\lceil \frac{c}{t} \right\rceil +r-1\le \Delta +r-1.$$ Note that for each  ${{F}_{i}},\text{ }1\le i\le n$, we must have $\left| {{F}_{i}} \right|\ge r+1$, as otherwise there will be a case when all servers containing the $i^{th}$ file are unavailable, which results in the permanent loss of this file. This implies that $\left| I \right|\ge r+1.$ 
Therefore, $$\left| \bigcup\limits_{i\in I}{{{S}_{i}}} \right|=c>(c+tr-1)-tr\ge t\left( \left| I \right|-r \right),$$ implying that the condition (2) of Theorem \ref{thr:main} does not hold. 

Supposing that $A$ does not satisfy condition $(2)$ of Theorem \ref{thr:main}. Then there exists $d\in \{r,\ldots ,r+\Delta -1\}$ and subset $I\subset \{1,2,\ldots ,m\}$ of size $d$ in which $Y=\bigcup\limits_{i\in I}{{{S}_{i}}}$ has more than $t(d-r)$ elements. Take $J\subset Y$ and $\left| J \right|=t(d-r)+1$. It is not difficult to check that
$$t(d-r)+1\le t(\Delta -1)+1<k.$$
Hence, $\left| J \right|<k$ and 
$$\left| \bigcup\limits_{j\in J}{{{F}_{j}}} \right|=d<(d-r+1)+r=\left\lceil \frac{t(d-r)+1}{t} \right\rceil +r,$$ which implies that the condition $(1)$ does not hold.

\begin{table*}
  \centering
\caption{Construction for ECBC when $n=17,m=5,k=10,t=3,r=1$ \\ (the empty cells contain $0$)}
  \begin{tabular}{|c|c|c|c|c|c|c|c|c|c|c|c|c|c|c|c|c|c|}
\hline
   & 1 & 2 & 3 & 4 & 5 & 6 & 7 & 8 & 9 & 10 & 11 & 12 & 13 & 14 & 15 & 16 & 17 \\ \hline
$S_1$ & 1  & 1  & 1  &    &    &    &    &    &    &     &     &     & 1   & 1   & 1   & 1   & 1   \\ \hline
$S_2$ & 1  & 1  & 1  & 1  & 1  & 1  &    &    &    &     &     &     &     &     &     & 1   & 1   \\ \hline
$S_3$ &    &    &    & 1  & 1  & 1  & 1  & 1  & 1  &     &     &     &     &     &     & 1   & 1   \\ \hline
$S_4$ &    &    &    &    &    &    & 1  & 1  & 1  & 1   & 1   & 1   &     &     &     & 1   & 1   \\ \hline
$S_5$ &    &    &    &    &    &    &    &    &    & 1   & 1   & 1   & 1   & 1   & 1   & 1   & 1   
\\ \hline
\end{tabular} 
\end{table*}

By combining what we have proved in the two steps, the theorem follows.
\end{proof}

We now proceed to explore the application of Theorem ~\ref{thr:main} in determining the optimal total storage for \ecbc{}.

\subsection{ECBC with Optimal Total Storage}

In is section, we use Theorem \ref{thr:main} to obtain ECBC with optimal total storage in some specific cases. 
The previous researches have shown that in almost all cases, the minimum total storage for an \ecbc{} where $r=0$ and $t=1$ can be represented as a function of its parameters. By replacing $k\to k+r$ for the case of $t=1,r>0$ or $k\to \left\lceil \frac{k}{t} \right\rceil$ for $t>1,r=0$, we can obtain the corresponding answers. Thus, for the general case of $t \ge 1,r \ge 0$, it is possible to replace $k \to \left\lceil \frac{k}{t} \right\rceil +r$ and apply the known results to find new bounds, the rest of work is construction. Let $N(n,m,k,t,r)$ denote the minimum total storage achievable by an \ecbc{}.

\begin{theorem}
If $m=\left\lceil \frac{k}{t} \right\rceil +r$ and $n\ge tm$ then $N(n,m,k,t,r) = m(n-tm+r+1).$
\end{theorem}

\begin{proof} By applying the condition $(2)$ of Theorem \ref{thr:main}, we can see that each set of size $d=m-1$ servers contains at most $t(m-1-r)$ files. As a result, the rest contains at least $n-t(m-1-r)$ files. Hence, by counting the number of storage on overall servers, we have
$$\begin{aligned} N & \ge m(n-t(m-1-r))=mn-tm(m-1-r) \\
&= tm(r+1) + m(n-tm). \end{aligned} $$

\textit{Construction}: each file from indices $1,2,\ldots,tm$ is stored on $r+1$ servers: files $1 \to t$ on servers $1 \to r+1$, files $t+1 \to 2t$ on servers $2 \to r+2$, and so on; the remaining $n-tm$ files are store on all servers. One can verify that in this case, $N(n,m,k,t,r)=tm(r+1)+m(n-tm)$ as needed and this also satisfies the ECBC condition.
\end{proof}

\textbf{Example.} For $k=10, t=3, r = 1$ then $m=5$ and $n = 17$ then $N_{\min} =17 \cdot 5-3 \cdot 5 \cdot 3=40.$ The construction is showed in Table I. \medskip

The bound for another case can be obtained as follows (the similar idea for construction can be found in in 
\cite{paterson2009combinatorial}):

\begin{theorem}
Let $h\triangleq \left\lceil \frac{k}{t} \right\rceil$. If $n\ge (h-1)\binom{m}{r+h-1}$ and $m\ge h+r$ then we have $$N(n,m,k,t,r)=nh-(h-1) \binom{m}{r+h-1}.$$
\end{theorem}

\subsection{Application to the Multiset CBC with $r>0$}

In~\cite{zhang2018multiset}, Zhang, Yaakobi, and Silberstein studied the multiset CBC where each file can appear multiple times in the retrieval list, i.e each $F_j \in \{1,2,\ldots,n \}$ is a multiset with elements having multiplicity up to $p$. Now we consider the erasure multiset ECBC in the presence of server failures. Based on Theorem \ref{thr:main}, we can state the Hall's conditions for this problem following the study in~\cite{zhang2018multiset}. 

\begin{theorem} Denote $\delta = \left\lceil k/p \right\rceil$. The sets $F_j$, $j \in \{1,2,\ldots,n\}$ represent a multiset ECBC with parameters $m,n,k,r,t,p$ if and only if one of the following conditions is satisfied.
\begin{enumerate}
\item[(1)] For every $c\in \{ 1,2,\ldots ,\delta \}$ and for every subset $J\subset \{1,2,\ldots ,n\}$ of size $c$, it holds that $$\left| \bigcup\limits_{j\in J}{{{F}_{j}}} \right|\ge \left\lceil \frac{\min \{cp,k\}}{t} \right\rceil +r.$$
\item[(2)] For every $d\in \{ r,r+1,\ldots ,r+\left\lceil \frac{p \delta}{t} \right\rceil -1 \},$ and for every subset $I\subset \{1,2,\ldots ,m\}$, it holds that
\[\left| \bigcup\limits_{i\in I}{{{S}_{i}}} \right|\le t(d-r).\]
\end{enumerate}
\end{theorem}

\subsection{Method to Retrieve Files from the Given ECBC Matrix}

Consider an ECBC matrix $A$ with parameters $n,m,k,t,r$ as defined. The goal is to find the list of servers containing the $k$ files to query, given a list of $r$ unavailable servers. To solve the problem, we use a bipartite graph $G=(V,E)$ and apply a maximum matching algorithm such as Hopcroft-Karp which run in the time complexity as $O(\sqrt{V} E)$. The procedure consists of two steps:
\begin{enumerate}
    \item Construct a sub-matrix from the ECBC matrix that has columns corresponding to the $k$ files to retrieve and rows corresponding to servers (excluding unavailable ones) that each has at least one of these files, meaning each row has at least one number $1$. The size of this sub-matrix is at most $(m-r) \times k.$
    \item Apply the Hopcroft–Karp algorithm to find the matching in the following two cases:
    \begin{itemize}
        \item If $t = 1$, find the maximum matching between columns and rows (ECBC property of the matrix ensures this will contain $k$ pairs).
        \item If $t > 1$, create $t$ copies of each row and then apply the algorithm as if $t = 1$. The chosen edges do not share a common vertex which can ensure that each server is accessed at most $t$ times. From these edges, one can find the list of needed servers and the corresponding files on those servers.
    \end{itemize}
\end{enumerate} 
Hence we can apply this procedure in the realistic systems which use the CBC/ECBC ideas.
\begin{table*}
\centering
\caption{Construction for ECBC when $n=9,k=3,m=5,r=2$ in the consecutive version}
\begin{tabular}{|c|c|c|c|c|c|c|c|c|c|}
\hline
   & $1$ & $2$ & $3$ & $4$ & $5$ & $6$ & $7$ & $8$ & $9$ \\ \hline
$S_1$ &  1  &    &   & 1  &   &  & 1 &    &  \\ \hline
$S_2$ &   &  1  &  &  &  1  &  &  &  1 &    \\ \hline
$S_3$ &    &  &  1  &   &  &  1  &  &  & 1    \\ \hline
$S_4$ & 1  & 1  & 1  & 1  & 1  & 1  & 1  & 1  & 1  \\ \hline
$S_5$ & 1  & 1  & 1  & 1  & 1  & 1  & 1  & 1  & 1  \\ \hline
\end{tabular}
\end{table*}

\section{Batch Codes with Consecutive Items}
In this section, we discuss a special version for the ECBC problem in which the retrieved files have consecutive indices. 
For examples, databases storing purchasing history, weather forecasts, school reports, medical records, can often be arranged in a linear order chronologically. A set of consecutive files provides relevant data for a specific period of time.
Although a quite simple solution exists for this problem, the setting in which only specially structured subsets of files (instead of arbitrary subsets like in the traditional batch code) are retrieved is novel and has potential applications in practice. The only other work in this direction that we are aware of is~\cite[App. F]{cao2023parallel}, in which the database items are placed at the nodes of a binary tree and only batches of items lying along a root-to-leaf path are retrieved, which correspond to Merkle proofs in a Merkle tree~\cite{Merkle1988}.    

\subsection{Solution for General ECBC}


If $t = 1$, we need the constraints $n\ge k$ and $m\ge k+r$. If $r$ servers are unavailable and there are fewer than $k+r$ servers, it is not possible to retrieve $k$ files.
Moreover, the lower bound on the number of servers storing each file is $r + 1$, which ensures all $r$ servers storing the same file are active. Therefore, we get the necessary condition $N\ge (r+1)n.$ \medskip

We demonstrate that the number $N$ is the minimum for all values of $m$ and $n$. We first construct the ECBC matrix:
\begin{enumerate}
\item[(1)] We store the files with the indices congruent to $i$ modulo $k$ on the servers with indices $i\in \{0,1,2,\ldots ,k-1\}$. This way, each file is stored once, resulting in a total storage count of $n$.
\item[(2)] For the servers with indices $i\in \{k,k+1,\ldots ,k+r-1\}$, we store all files. So the total storage count to $rn$.
\end{enumerate}

Hence, the total number of storage is ${{N}_{\min }}=n+rn=(r+1)n.$ It is not difficult to check that this construction satisfies all the given ECBC conditions:
\begin{itemize}
\item If $r$ servers in group (2) are unavailable, we still have servers in group (1) to retrieve the files. Note that we only consider $k$ consecutive files, two files with indices that have a difference greater than $k$ will not be retrieved simultaneously. This is the main idea of the construction.
\item If some servers in group (1) is unavailable, we can replace it with any server in group (2), since each server in that group stores all files. 
\end{itemize}

For the case where $t>1$, the above result remains tight, since the bound  $(r+1)n$ holds for all values of $t$. It just enables us to reduce the number of servers from $k$ to $\left\lceil \frac{k}{t} \right\rceil $.

So, in the general ECBC problem, the solution for consecutive files version is ${{N}_{\min }}=(r+1)n.$

\textbf{Example.} For $n=9, m=5, k=3, r = 2$ then $N_{\min} =(r+1)n = 27.$ The construction is showed on the Table II. \medskip

In this case, by the periodically distribution of the files among servers, one can modify the given matrix $A$ when some of values $m, n$ or $r$ change as follows: if the number $n$ of files increases then add new columns to the matrix such that for each of them, the last $r$ cells are all filled by $1$, an extra number $1$ will be filled base on the modulo $k$; if the number $m$ of servers increases, it does not matter since $k+r$ servers are enough for the system; and if the upper bound $r$ increases then the number of servers must increases too, each of the new servers will store all files there.
When number $k$ or $t$ changes, it is required to edit quite a lot to get the new optimal system.

\subsection{The Relations between CBC/ECBC and NAGT}


NAGT is a technique that aims to minimize the number of tests required to identify some particular defected items among all items. The items are grouped and tested together to improve the accuracy. Tests are performed independently, can be done in parallel, and are represented by a measurement matrix.

In~\cite{jia2019erasure}, Jia {\it et al.} studied the connections between ECBC and NAGT through the incident binary matrices, focusing on disjunct and separable matrices. In \cite{colbourn1999group,bui2021improved}, NAGT with consecutive positives has been studied, where positive items are consecutive in a linear order. This idea inspired the ECBC problems, where the files are also linearly ordered and the retrieval files are consecutive. In this special aspect, the optimal value of storage in the ECBC problem can easily be obtained as the previous part. \medskip

Furthermore, NAGT is a well-studied problem with numerous theories and research results have been done. By discovering new relationships between two problems ECBC and NAGT, we can potentially apply ideas from one to another. Here, we compare some other aspects of these two problems.


\begin{itemize}
\item Objectives: ECBC aims to store $n$ files in $m$ servers, while NAGT aims to test $N$ items with $T$ tests.
\item Requested items: ECBC requires $k$ files (a subset of the original files), while NAGT requires $D$ defect items (a subset of the original items).
\item The constrains: in ECBC, it is required $k \le \min\{ m,n\}$; in NAGT, it is required $D \le \min\{N,T\}$.
\item Classification of storage objects: ECBC classifies servers into two groups - those that contain at least one required file and others. In NAGT, tests are classified into positive class and negative class.
\item Erasure aspects: in ECBC, there are at most $r$ unavailable servers while in NAGT, there are $R$ inhibitors that can join some tests and change the result of tests.
\end{itemize}
Despite some similar aspects of these two problems, the main difference lies in the finding of subsets. ECBC finds a subset of servers based on the storage of $k$ files, while NAGT finds a subset of items based on the test results.

\section{Conclusions}
In this work, we extended the previous results on Erasure Combinatorial Batch Codes to accommodate the most general case when $r\geq 0$ and $t \geq 1$, using an generalization of Hall's theorem.
We also explored the relationship between ECBC and NAGT problems, and found an optimal solution for a special case when retrieved files are arranged in specific consecutive linear sequence. This provides new insights into the problem and a new approach to solve it.

In the future, we will consider the way to slightly modify the given ECBC matrix $A$ in case that some of values $n,m,k,t,r$ change (since the modification of whole matrix requires plenty of computation). We also aim to explore the practical applications of ECBC in real-time data storage for IoT devices and multiple servers. Our research also has significant potential in industries requiring reliable data storage, such as biomedical and agriculture, thus we would like to further investigate the ECBC concept and its potential use in these industries.

\section{Acknowledgments}
\label{sec:ack}

This article was funded in part by University of Science, VNU-HCM under Grant No. CNTT2022$-$11.

We also would like to express our gratitude to Mr. Bui Van Thach (National University of Singapore) who has shared with us a lot of valuable ideas to complete this research.

\bibliographystyle{ieeetr}
\balance
\bibliography{bibli}

\end{document}